\documentclass[copyright,creativecommons]{eptcs}
\providecommand{\event}{AiML 2026} 

\usepackage{aiml26}

\usepackage{iftex}

\ifpdf
  \usepackage{underscore}         
  \usepackage[T1]{fontenc}        
\else
  \usepackage{breakurl}           
\fi

\usepackage{amsmath,amssymb,stmaryrd}

\renewcommand{\phi}{\varphi}
\newcommand{\slg}{\mathbf} 
\newcommand{\kf}{\mathfrak}  
\newcommand{\cl}{\mathcal} 
\newcommand{\set}[2]{
  \left\{ #1 \,:\, #2  \right\}%
} 
\newcommand{\PV}{\mathrm{Prop}}

\makeatletter
\newcommand{\dotop}[1]{{\mathpalette\dotop@{#1}}}
\newcommand{\dotop@}[2]{%
  \vphantom{#2}%
  \ooalign{$\m@th#1\mathop#2$\cr\hidewidth$\m@th#1\cdot$\hidewidth\cr}%
}
\makeatother
\newcommand{\Diamd}{{\mathop\lozenge}}
\newcommand{\Boxd}{{\mathop\square}}
\newcommand{\Boxc}{\dotop{\square}}

\newcommand{\Lang}{\cl{L}}
\newcommand{\Langc}{\cl{L}_c}
\newcommand{\fval}[1]{\llbracket #1 \rrbracket}
\newcommand{\eps}{\varepsilon}
\DeclareMathOperator{\Sub}{Sub} 
\newcommand{\Eq}{\Longleftrightarrow}
\newcommand{\refcl}[1]{\widehat{#1}} 
\newcommand{\proj}{g}

\usepackage{enumitem} 

\newcommand{\NewCase}{\vspace{\baselineskip}\noindent} 

\RequirePackage{array}
\newcommand\estyle{}%
\newenvironment{authors}[1]{%
  \begingroup
  \renewcommand\institute[1]{%
    \\\multicolumn{#1}{@{}c@{}}{\scriptsize
        \begin{tabular}[t]{@{}>{\footnotesize}c@{}}%
          ##1%
        \end{tabular}%
      }%
  }%
  \renewcommand\email[1]{%
    \gdef\estyle{\footnotesize\ttfamily}%
    \\##1%
    \gdef\estyle{}%
  }
  \begin{tabular}[t]{@{}*{#1}{>{\estyle}c}@{}}%
}{%
  \end{tabular}%
  \endgroup
}
\title{Topological Logics of Path-Reachability}
\author{%
  \begin{authors}{2}
    Aleksandr Gagarin & David Fern\'{a}ndez-Duque
    \institute{University of Barcelona, Spain}
    \email{agagarga7@alumnes.ub.edu & fernandez-duque@ub.edu}
  \end{authors}%
}
\newcommand{\titlerunning}{Topological Logics of Path-Reachability}
\newcommand{\authorrunning}{A. Gagarin \& D. Fern\'{a}ndez-Duque}

\hypersetup{
  bookmarksnumbered,
  pdftitle    = {\titlerunning},
  pdfauthor   = {\authorrunning},
  pdfsubject  = {Topological Logics of Path-Reachability},
  pdfkeywords = {spatial logic, topological semantics, cantor derivative, neighborhood semantics} 
}

\begin{document}
\maketitle

\begin{abstract}
The topological semantics of modal logic has been an active area of research ever since their introduction in the 1940s, with attention shifting in recent years from standard unimodal logic to more expressive frameworks. In particular, an Until-like path-reachability modality has recently been studied in Bezhanishvili et al. (2024) in polyhedral semantics; this paper investigates its topological counterpart. Focusing on the language combining said modality with the classical Cantor derivative modality, we exhibit an axiomatic system sound and complete both for the class of $T_1$ topologies and for the class of all metric spaces, and establish its decidability. 
We also axiomatize the logic of all topological models in a weaker language obtained by substituting the closure modality for the Cantor derivative. To prove our results, we introduce an equivalent neighborhood-like semantics allowing for the finite model property. 
\end{abstract}

\section{Introduction}

Topological semantics provide ways of interpreting modal formulas in spatial terms, by viewing propositions as regions of a topological space and modal operators as expressing some spatial transformations. Two classical variants of topological semantics, both introduced by McKinsey and Tarski~\cite{MT44}, are the \emph{c-semantics} and the \emph{d-semantics}.
In the c-semantics, the modal operators $\Diamond$ and $\Box$ are interpreted as closure and interior, respectively. The logic of all topological spaces then coincides with $\slg{S4}$. Moreover, $\slg{S4}$ is the logic of any dense-in-itself metric space~\cite{MT44, BB07}. 
In the more expressive d-semantics, the $\Diamond$ modality is interpreted as the Cantor derivative, which maps a region to the set of its limit points. The logic of all topologies is $\slg{wK4}$~\cite{BB07}. The separation axiom $T_D$ is modally definable by $\Box p \to \Box^2 p$, and $\slg{K4}=\slg{wK4}+\Box p \to \Box^2 p$ is the logic of metrizable Stone spaces~\cite{BEG10}.

An Until-like binary modality $\gamma$ for path-reachability was introduced in~\cite{BCGGLM22} in the context of so-called \emph{polyhedral models}. The formula $\gamma(\phi,\psi)$ is defined to be true at a point $x$ iff there exists a continuous path beginning at $x$, with all intermediate points validating $\phi$, and the end point validating $\psi$. 
In~\cite{BBCFG24}, the logic of polyhedral models as well as the logic of Alexandroff spaces with this modality were axiomatized and shown to be decidable. 
In polyhedral and Alexandroff models, $\gamma(\phi,\top)$ defines the closure of the set defined by $\phi$, thus $\gamma$-semantics subsumes c-semantics~\cite{BBCFG24}. 
This is not necessarily the case in an arbitrary topology, e.g., in any non-trivial totally path-disconnected space.

In this paper, we study logics related to $\gamma$ in a general topological setting. In the language combining $\gamma$ with c-semantics, we axiomatize the logic of all topologies and show that it coincides with the logic of metric spaces. In the language combining $\gamma$ with d-semantics, we exhibit a formula defining the separation axiom $T_1$, axiomatize the logic of all $T_1$ topologies, and show that it is complete for the class of metric spaces. Thus, both logics are also characterized by the class of Hausdorff spaces, by the class of regular Hausdorff spaces, etc. Further, we establish that these logics are decidable.

In proving our results, we introduce an equivalent neighborhood-like semantics (cf.~\cite{Pacuit}) for $\gamma$, where the set of intermediate points of a path is treated as a ``neighborhood'' of the path's original point. 
Then we establish the finite model property in this semantics by constructing an appropriate filtration of the canonical neighborhood-like model. In Sections~\ref{S:preliminaries}--\ref{S:unraveling}, we obtain our result for the d-semantics, while their c-semantical counterparts are derived in Section~\ref{S:interiorModality}.

\section{Preliminaries}\label{S:preliminaries}

In this section we establish some of the notation we will use. The powerset of $X$ is denoted $\mathcal{P}(X)$. For a transitive binary relation $R$ on a set $X$, we write $R(x) := \set{y \in X}{x\mathrel{R}y}$, and call
$U \subseteq X$ an \emph{$R$-upset} if $R(x) \subseteq U$ for all $x \in U$.

\textbf{Topology.} 
We use the following standard notions from general topology; see e.g.~\cite{Engelking} for definitions: topological space; open, closed, discrete, and connected subsets; continuous map; metric space and its induced topology.
A topological space $\langle X,\tau\rangle$ is called \emph{$T_1$} if for all distinct $x,y \in X$ there exists an open set $U$ with $x \in U$ and $y \not\in U$. Note that all metric spaces induce $T_1$ topologies~\cite{Engelking}.

When a subset of $\mathbb{R}$ is mentioned (e.g., $[0,1]$ or $\mathbb{Q}$), the usual metric given by $d(x,y):=|x-y|$ and its induced topology are implicit. 
A \emph{path} in a topological space $\langle X,\tau\rangle$ is a continuous map $[0,1] \to X$. 

\textbf{Miscellaneous.} 
For a map $\pi:[0,1] \to X$ and $a,b \in [0,1]$, denote by $\pi \restriction [a,b]$ the map $[0,1]\to X$ given by $t \mapsto \pi(bt-at+a)$. 
Observe that if a topology on $X$ is given and $\pi$ is a path, then $\pi \restriction [a,b]$ is also a path.
Note that we allow $b<a$; in particular, $\pi \restriction [1,0]$ is the ``reverse'' of $\pi$.  

For a set $X$ and $x,y \in X$, denote by $\iota_{x,y}$ the map $[0,1] \to X$ given by $\iota_{x,y}(0)=x$ and $\iota_{x,y}(t)=y$ for $t \in (0,1]$. In particular, $\iota_{x,x}$ is constant. Note that $\iota_{x,y}$ may or may not be continuous when some topology on $X$ is given.

\textbf{Language.}
Fix a countably infinite set of propositional variables $\PV$. Denote by $\Lang$ the language built from $\PV$ and $\bot$ using the unary connective $\Boxd$ and binary connectives $\to$ and $\gamma$. 
For a formula $\phi\in\Lang$, denote by $\Sub\phi$ the set of its subformulas.

\textbf{Semantics.} 
We interpret $\Box$ as in the usual topological $d$-semantics~\cite{BB07,BEG10}, and $\gamma$ as the reachability modality from~\cite{BBCFG24, BCGGLM22}. Formally, a \emph{topological model} is a triple $\kf{M}=\langle X,\tau,\fval{\cdot}\rangle$, where $\langle X, \tau\rangle$ is a topological space, and $\fval{\cdot} : \PV \to \cl{P}(X)$. The map $\fval{\cdot}$ is extended to $\Lang$ as follows:
\begin{itemize}

	\item $\fval{\bot}=\varnothing$;
	\item $\fval{\phi \to \psi} = (X \setminus \fval{\phi}) \cup \fval{\psi}$;
	\item $x \in \fval{\Boxd \phi}$ iff there exists $U \in \tau$ such that $x \in U$ and $U \setminus \{x\} \subseteq \fval{\phi}$; 
	\item $x \in \fval{\gamma(\phi,\psi)}$ iff there exists a path $\pi$ with $\pi(0) = x$, $\pi[(0,1)] \subseteq \fval{\phi}$, and $\pi(1) \in \fval{\psi}$. 
\end{itemize}
We write $\kf{M} \models \phi$ iff $\fval{\phi}=X$.

\textbf{Abbreviations.}
We use standard abbreviations $\top$, $\neg$, $\vee$, and $\wedge$. We also abbreviate $\Boxc \phi := \phi \wedge \Boxd \phi$ and $\widehat\gamma(\phi,\psi) := \phi\wedge \gamma(\phi,\phi\wedge\psi)$.
Note that in a topological model:
\begin{itemize}
	\item $\fval{\Boxc\phi}$ is the interior of $\fval{\phi}$. In particular, all $\fval{\Boxc\phi}$ are open sets;
	\item $x \in \fval{\widehat\gamma(\phi,\psi)}$ iff there exists a path $\pi$ with $\pi(0)=x$, ${\pi[[0,1]] \subseteq \fval{\phi}}$, and $\pi(1) \in \fval{\psi}$. 
\end{itemize}

\section{Axiomatization}

In this section, we present an $\Lang$-formula defining the separation axiom $T_1$. We also exhibit an axiomatic system $\slg{TLR}$ and verify its soundness for $T_1$ topologies. 

\begin{proposition}\label{P:t1Characterization}
A topological space $\langle X,\tau\rangle$ is $T_1$ if and only if $\gamma(p\wedge \Boxd \neg p, \top) \to p$ is valid in all topological models based on $\langle X,\tau\rangle$.
\end{proposition}
\begin{proof}
Suppose $\langle X,\tau\rangle$ is not $T_1$, i.e., there exist distinct $x,y\in X$ such that every open neighborhood of $x$ contains $y$, and thus $\iota_{x,y}$ is continuous. Take $\fval{p} := \{y\}$.
Now $y \in \fval{\Boxd \neg p}$, whence $\iota_{x,y}$ witnesses $x \in \fval{\gamma(p\wedge \Boxd \neg p,\top)}$. However, $x \not\in \fval{p}$.

Conversely, suppose $\langle X,\tau\rangle$ is $T_1$. Consider a valuation $\fval{\cdot}$ and a point $x$ with $x \in \fval{\gamma(p\wedge\Boxd\neg p,\top)}$ witnessed by some path $\pi$. The set $\fval{p\wedge \Boxd \neg p}$ is discrete, whence its connected subset $\pi[(0,1)]$ is some singleton $\{y\}$. Now, $\pi$ is continuous with $\pi(0)=x$ and $\pi[(0,1)]=y$, thus every open neighborhood of $x$ contains $y$, 
hence $x=y \in \fval{p}$. 
\end{proof}

\begin{definition}
The logic $\slg{TLR}$ is defined by all axioms and rules of $\slg{K4}$ for $\Boxd$ together with the following:
\begin{itemize}
	\item[(A0)] $\gamma(\phi, \psi_1 \vee \psi_2) \to \gamma(\phi, \psi_1) \vee\gamma(\phi,\psi_2)$;
	\item[(A1)] $\gamma(\phi, \phi \wedge \gamma(\phi,\psi)) \to \gamma(\phi,\psi)$;
	\item[(A2)] $\phi \to \gamma(\phi,\phi)$;
	\item[(A3)] $\neg \gamma(\bot,\top)$;
	\item[(A4)] $\gamma (\phi, \neg \gamma(\phi,\psi)) \to \neg \psi$;
	\item[(A5)] $\gamma(\phi, \psi) \to \gamma(\phi \wedge \gamma(\phi,\psi),\psi)$;
	\item[(A6)] $\Boxc \psi \wedge \gamma(\phi,\chi) \to \gamma(\phi \wedge\Boxc\psi, (\phi \wedge \neg \Boxc\psi) \vee\chi)$; 
	\item[(A7)] $\widehat\gamma\big(\chi\wedge (\Boxc\phi_1\vee\Boxc\phi_2)\wedge(\widehat\gamma(\chi\wedge\Boxc\phi_1,\psi)\vee \widehat\gamma(\chi\wedge\Boxc\phi_2,\psi)\to\psi), \psi\big) \to \psi$; 
	\item[(A8)] $\gamma(\phi \wedge \Boxd \neg \phi, \top) \to \phi$;
	\item[(Mon)] $\begin{array}{r} \phi\to\phi' \quad \psi\to\psi' \\ \cline{1-1} \gamma(\phi,\psi)\to\gamma(\phi',\psi')\\ \end{array}$.
\end{itemize}
\end{definition}

Now we verify soundness using some natural properties of paths in a topology; e.g., (A1) expresses that a concatenation of two paths is also a path, (A2) signifies the existence of constant paths, etc. 

\begin{proposition}\label{P:soundness}
The logic $\slg{TLR}$ is sound for the class of $T_1$ topologies. Moreover, (A0)--(A7) and (Mon) are valid in all topologies.
\end{proposition}
\begin{proof}
Consider a topological model $\kf{M}=\langle X, \tau,\fval{\cdot}\rangle$ and a point $x \in X$.

(A0): If $x \in\fval{\gamma(\phi,\psi_1\vee\psi_2)}$ is witnessed by a path $\pi$, then $\pi(1) \in \fval{\psi_1 \vee \psi_2}$, hence  $\pi(1) \in \fval{\psi_i}$ for some $i \in \{1,2\}$, whence $\pi$ witnesses $x \in\fval{\gamma(\phi,\psi_i)}$.

(A1): Suppose $x \in \fval{\gamma(\phi,\phi\wedge \gamma(\phi,\psi))}$ is witnessed by $\pi_1$, and $\pi_1(1) \in\fval{\gamma(\phi,\psi)}$ is witnessed by $\pi_2$. Now the path $\pi$ given by $\pi \restriction [0,\frac12]=\pi_1$ and $\pi \restriction [\frac12,1]=\pi_2$ witnesses $x \in \fval{\gamma(\phi,\psi)}$.

(A2): If $x \in \fval{\phi}$, then $\iota_{x,x}$ witnesses $x \in \fval{\gamma(\phi,\phi)}$.

(A3) is valid since $\pi[(0,1)] \neq\varnothing$ for every path $\pi$.

(A4): Suppose $x \in \fval{\gamma(\phi,\neg\gamma(\phi,\psi))}$ is witnessed by a path $\pi$, and assume $x \in\fval{\psi}$. Then $\pi \restriction [1,0]$ witnesses that $\pi(1) \in \fval{\gamma(\phi,\psi)}$, which is a contradiction.

(A5): Suppose $x \in \fval{\gamma(\phi, \psi)}$ is witnessed by a path $\pi$. Then for every $t \in (0,1)$, the path $\pi \restriction [t,1]$ witnesses $\pi(t) \in \fval{\gamma(\phi,\psi)}$. Therefore, $\pi$ witnesses  $x \in \fval{\gamma(\phi\wedge \gamma(\phi,\psi), \psi)}$.

(A6): Suppose $x \in\fval{\Boxc \psi}$ and $x \in \fval{\gamma(\phi,\chi)}$ is witnessed by a path $\pi$. If $\pi[(0,1)] \subseteq \fval{\Boxc \psi}$, then $\pi$ witnesses  $x \in \fval{\gamma(\phi\wedge \Boxc \psi, \chi)}$. Otherwise, denote by $t$ the least element of the closed set $\pi^{-1}[\fval{\neg \Boxc\psi}]$. As $0<t<1$, 
the path $\pi \restriction [0,t]$ witnesses that $x \in \fval{\gamma(\phi\wedge \Boxc \psi, \phi \wedge \neg \Boxc\psi)}$.

(A7): Suppose $x \in \fval{\widehat\gamma(\chi\wedge (\Boxc\phi_1\vee\Boxc\phi_2)\wedge(\widehat\gamma(\chi\wedge\Boxc\phi_1,\psi)\vee \widehat\gamma(\chi\wedge\Boxc\phi_2,\psi)\to\psi), \psi)}$ is witnessed by a path $\pi$, and assume $x \not\in\fval{\psi}$. Denote by $t$ the supremum of $\pi^{-1}[\fval{\neg\psi}]$. Fix $i \in \{1,2\}$ such that $\pi(t) \in \fval{\Boxc\phi_i}$ and $\eps>0$ such that $[\max\{t-\eps,0\},\min\{t+\eps,1\}] \subseteq \pi^{-1}[\fval{\Boxc\phi_i}]$. 
Now $\pi \restriction [t,\min\{t+\eps,1\}]$ witnesses that $\pi(t) \in \fval{\widehat\gamma(\chi\wedge\Boxc\phi_i, \psi)}$.
As $\pi[[0,1]] \subseteq \fval{\widehat\gamma(\chi\wedge \Boxc\phi_i, \psi) \to \psi}$, we conclude that $\pi(t) \in \fval{\psi}$. By the choice of $t$, there exists $t' \in \pi^{-1}[\fval{\neg \psi}] \cap [t-\eps,t]$, thus $\pi \restriction [t',t]$ witnesses $\pi(t') \in \fval{\widehat\gamma(\chi\wedge \Boxc\phi_i,\psi)}$, hence $\pi(t') \in \fval{\psi}$, which is a contradiction. 

(Mon): If $\fval{\phi}\subseteq \fval{\phi'}$ and $\fval{\psi}\subseteq\fval{\psi'}$, then $\fval{\gamma(\phi,\psi)}\subseteq\fval{\gamma(\phi',\psi')}$.

If $\langle X,\tau\rangle$ is $T_1$, then the validity of (A8) follows from Proposition~\ref{P:t1Characterization}, and the validity of $\slg{K4}$ is well-known~\cite{BB07}.
\end{proof}

\section{Neighborhood-like Semantics}\label{S:canonical}

In this section, we define an alternative semantics for $\Lang$ (\emph{flanked Kripke frames}) and a certain class of finite frames in it (\emph{suitable frames}). In the following sections, we will establish our results by proving that $\slg{TLR}$ is complete for suitable frames (Section~\ref{S:fmp}) and that suitable frames define the same logic as $T_1$ topologies (Section~\ref{S:unraveling}). 

\begin{definition}
A \emph{flanked Kripke frame} is a triple $\langle X, R, P \rangle$, where:
\begin{itemize}
	\item $R \subseteq X \times X$ is transitive; 
	\item $P \subseteq X \times \cl{P}(X) \times X$ is \emph{monotone}, i.e., if $\langle x, S,y\rangle \in P$ and $S \subseteq S' \subseteq X$, then $\langle x,S',y\rangle\in P$.
\end{itemize}
A \emph{flanked Kripke model} is a flanked Kripke frame equipped with a valuation $\fval{\cdot} : \PV \to \cl{P}(X)$, which is extended to $\Lang$ as follows: 
\begin{itemize}
	\item $\fval{\bot} = \varnothing$; 
	\item $\fval{\phi\to\psi} = (X \setminus \fval{\phi}) \cup \fval{\psi}$;
	\item $x \in\fval{\Boxd\phi}$ iff $R(x) \subseteq \fval{\phi}$;
	\item $x \in \fval{\gamma(\phi,\psi)}$ iff there exists $y \in \fval{\psi}$ satisfying $\langle x,\fval{\phi},y \rangle \in P$.
\end{itemize}
\end{definition}

\begin{remark}\label{R:semanticMotivation}
This semantics combines a Kripke structure $R$ for $\Box$ and and a ``flanked'' structure $P$ for $\gamma$. The former is natural to consider as $\slg{K4}$ is the logic of $T_1$ topologies~\cite{BB07}. To motivate the latter, note that the topological interpretation of $\gamma$ may be rewritten as follows: $x \in \fval{\gamma(\phi,\psi)}$ iff there exists $y \in \fval{\psi}$ with $\langle x,\fval{\phi},y\rangle \in P$, where $P$ is a monotone set defined as follows:
$$
	P := \set{\langle x, S,y\rangle \in X \times \mathcal{P}(X) \times X}{\text{there exists a path $\pi$ with $\pi(0)=x$, $\pi[(0,1)]\subseteq S$, and $\pi(1)=y$}}.
$$
 \end{remark}

This ``flanked'' semantics for $\gamma$ is similar to the monotone neighborhood semantics used for non-normal modal logics such as $\slg{EM}$ (cf.~\cite{Pacuit}). 

We now list the frame properties that roughly correspond to the axioms (A1)--(A8). 

\begin{definition}
A flanked Kripke frame $\langle X,R,P\rangle$ is \emph{suitable} if it is finite and has the following 8 properties:
\begin{itemize}
	\item[(F1)] if $\langle x, S,z\rangle \in P$, $\langle z,S,y\rangle \in P$, and $z \in S$, then $\langle x, S,y\rangle \in P$; 
	\item[(F2)] $\langle x, \{x\}, x \rangle \in P$ for all $x \in X$;
\end{itemize}
and for all $\langle x,S,y\rangle \in P$:
\begin{itemize}
	\item[(F3)]  $S \neq \varnothing$;
	\item[(F4)] $\langle y,S,x\rangle \in P$;
	\item[(F5)] for each $z \in S$, either $\langle x, S \setminus \{z\}, y \rangle \in P$ or $\langle z, S, y\rangle \in P$; 
	\item[(F6)]  for every $R$-upset $U$ with $x \in U$, there exists $z \in (S \setminus U) \cup \{y\}$ with $\langle x,S \cap U, z\rangle \in P$; 
	\item[(F7)] if $x,y\in S$, then, for all $R$-upsets $U_1$ and $U_2$ with $S \subseteq U_1 \cup U_2$, there exists an $\langle S,U_1,U_2,x,y\rangle$-sequence. Here an \emph{$\langle S,U_1,U_2,x,y\rangle$-sequence} is a sequence $x_0,\dotsc,x_k$ with $k \geq 0$, $x_0=x$, and $x_k=y$ such that for each $i<k$ there exists $j_i \in \{1,2\}$ satisfying $x_i,x_{i+1} \in S \cap U_{j_i}$ and $\langle x_i, S \cap U_{j_i}, x_{i+1} \rangle \in P$; 
	\item[(F8)] if $S=\{y\}$, then either $x = y$ or $yRy$. 
\end{itemize}
\end{definition}

In the remainder of this section, we briefly motivate properties (F1)--(F8) and exhibit some illustrative examples.

\begin{remark}
Note that $R$-upsets form an Alexandroff topology on $X$ (cf.~\cite{BB07}).
Given a topological space, the set $P$ from Remark~\ref{R:semanticMotivation} satisfies (F1)--(F5), as well as (F6) and (F7) if we substitute ``open set'' for ``$R$-upset.'' One can verify this similarly to the proof of Proposition~\ref{P:soundness}, invoking the same topological properties.
\end{remark}

\begin{remark}
All axioms and rules of $\slg{K4}$, axiom (A0), and rule (Mon) are valid in all flanked Kripke frames.
One can check that (F1)--(F4) and (F6)--(F7) are precisely the frame properties corresponding to (A1)--(A4) and (A6)--(A7). However, (F5) is defined by (A5) only over finite frames, and (F8) is defined by (A8) only over frames satisfying (F3)--(F6). We employ (F5) and (F8) merely for convenience.  
\end{remark}

\begin{example}\label{E:positive}
Take $X=\{a,b,c\}$, $R=(\{b,c\} \times X) \cup \{\langle a,a\rangle\}$, and $\langle x,S,y\rangle\in P$ iff (i) $x=y \in S$, (ii) $S=X$, or (iii) $a \in S$ and $x,y \in \{a,b\}$. 
A direct inspection shows that $\langle X,R,P\rangle$ is suitable. 
If we omit $\langle a,a\rangle$ from $R$, the resulting frame refutes (F8) since $\langle b,\{a\},a\rangle \in P$, but (F1)--(F7) are unaffected. 
If we omit (iii) from $P$, the resulting frame satisfies (F1)--(F5), (F7), and (F8), but not (F6). 
\end{example}

\begin{example}
Take $X = \{a,b,c,d,e,f\}$, $R = \{a,b\}^2 \cup \{c,d\}^2 \cup (X \times \{e,f\})$, and $\langle x,S,y\rangle \in P$ iff
(i) $x=y \in S$,
(ii) $x,y \in \{a,c,e\}$ and $e \in S$,
(iii) $x,y \in \{b,d,f\}$ and $f \in S$, or
(iv) $S=X$.
This $\langle X,R,P\rangle$ refutes (F7) since $\langle a,X,f\rangle \in P$ but there are no $\langle X,\{a,b,e,f\}, \{c,d,e,f\}, a,f\rangle$-sequences.
However, one can verify (F1)--(F6) and (F8). 
\end{example}

\begin{example}\label{E:allTopologies}
The axiomatic system obtained from $\slg{TLR}$ by dropping (A8) is sound but not complete for the class of $T_D$ spaces. Indeed, one can verify that $\phi := \gamma(p\wedge \Box\neg p\wedge \gamma(q,\top), \top) \to \gamma(q,\top)$ is valid in all topological models. 
Now take $X=\{a,b,c\}$, $R=\{\langle a,b\rangle, \langle b,c\rangle, \langle a,c\rangle\}$, and $\langle x,S,y\rangle \in P$ iff (i) $x=y \in S$, (ii) $x,y \in \{a,b\}$ and $b \in S$, (iii) $x,y \in \{b,c\}$ and $c \in S$, or (iv) $x,y \in \{a,b,c\}$ and $\{b,c\}\subseteq S$. As seen by inspection, $\langle X,R,P\rangle$ satisfies (F1)--(F7) and thus validates said axiomatic system. Taking $\fval{p} := \{b\}$ and $\fval{q} := \{c\}$, we obtain $a \not\in \fval{\phi}$, whence $\phi$ is not derivable.
\end{example}

\section{Finite Model Property}\label{S:fmp}

In this section, we verify that $\slg{TLR}$ has the finite model property in the flanked Kripke semantics, i.e., that $\slg{TLR}$ is complete for the class of suitable frames (Proposition~\ref{P:flankedFMP}).
First, we define the canonical flanked Kripke model, similarly to canonical models in monotone neighborhood semantics (cf.~\cite{Pacuit}). Then we construct a suitable filtration. 

\begin{definition}
A set $\Gamma \subseteq \Lang$ is \emph{$\slg{TLR}$-consistent} if there exists no finite $\Gamma' \subseteq \Gamma$ with $\slg{TLR} \vdash \bigwedge \Gamma' \to \bot$. 
The \emph{canonical flanked Kripke model} is $\kf{M}_c = \langle X_c, R_c, P_c, \fval{\cdot}_c\rangle$, where:
\begin{itemize}
	\item $X_c$ is the class of all maximal $\slg{TLR}$-consistent sets;
	\item $\Gamma \mathrel{R_c} \Delta$ iff for all $\Boxd \phi \in \Gamma$ we have $\phi \in \Delta$;
	\item $\langle \Gamma, S, \Delta\rangle \in P_c$ iff for all $\phi \in \bigcap S$ and $\psi \in \Delta$ we have $\gamma(\phi,\psi) \in \Gamma$; 
	\item $\Gamma \in \fval{p}_c$ iff $p \in \Gamma$.
\end{itemize}
\end{definition}

\begin{lemma}
$\kf{M}_c$ is indeed a flanked Kripke model.
\end{lemma}
\begin{proof}
Since $\slg{TLR} \vdash \Boxd\phi\to\Boxd^2\phi$, the relation $R_c$ is transitive. It should be clear that $P_c$ is monotone.
\end{proof}

As for the usual neighborhood semantics of~\cite{Pacuit}, we have the following:

\begin{lemma}\label{L:canonicalModelTruth}
For all $\Gamma\in X_c$ and $\chi \in \Lang$, we have $\Gamma \in \fval{\chi}_c$ iff $\chi \in \Gamma$.
\end{lemma}
\begin{proof}
By induction on $\chi$. The base case and the step for $\chi=\phi\to\psi$ or $\chi=\Boxd\phi$ are as in the standard canonical Kripke model construction for modal logic (cf.~\cite{CZbook});
we only consider $\chi=\gamma(\phi,\psi)$.
If $\Gamma \in \fval{\gamma(\phi,\psi)}_c$, then $\langle \Gamma, \fval{\phi}_c, \Delta\rangle \in P_c$ for some $\Delta \in \fval{\psi}_c$, whence $\phi \in \bigcap \fval{\phi}_c$ and $\psi \in \Delta$ by the induction hypothesis, hence $\gamma(\phi,\psi)\in \Gamma$ by definition of $P_c$.

Conversely, suppose $\gamma(\phi,\psi) \in \Gamma$. Assume $\Delta_0 := \{\psi\} \cup \set{\neg\psi' }{ \gamma(\phi,\psi') \not\in \Gamma}$ is not $\slg{TLR}$-consistent, witnessed by $\vdash \psi \to \bigvee_{i=1}^k \psi'_i$.
If $k>0$, it follows by (Mon) and (A0) that $\vdash \gamma(\phi,\psi) \to \bigvee_{i=1}^k \gamma(\phi, \psi'_i)$, which is a contradiction. If $k=0$, i.e., $\vdash \neg\psi$, then $\vdash \psi \to \neg\gamma(\phi,\top)$ and $\vdash \gamma(\phi,\neg\gamma(\phi,\top))\to \bot$ by (A4), hence $\vdash \gamma(\phi,\psi)\to\bot$ by (Mon), which is also a contradiction. Therefore, $\Delta_0$ is $\slg{TLR}$-consistent. Fix some $\Delta \in X_c$ extending $\Delta_0$. Now for every $\phi' \in \bigcap \fval{\phi}_c$ and $\psi' \in \Delta$ we have $\vdash \phi \to \phi'$ and $\gamma(\phi,\psi') \in \Gamma$, thus $\gamma(\phi',\psi')\in\Gamma$ by (Mon). Therefore, $\langle \Gamma,\fval{\phi}_c,\Delta \rangle \in P_c$ witnesses $\Gamma \in \fval{\gamma(\phi,\psi)}_c$.
\end{proof}

For the remainder of this section, fix a formula $\phi_0 \in \Lang$. 
To establish the finite model property, we take the ``closure under concatenations'' (i.e., under (F1)) of the ``minimal $\phi_0$-filtration.'' We first verify that it admits a ``filtration lemma'' (Lemma~\ref{L:filtrationLemma}). Then we show that it is a suitable frame, by establishing that the ``minimal $\phi_0$-filtration'' itself satisfies (F2)--(F8) (Lemma~\ref{L:minimalFiltrationProperties}) and that said ``closure under concatenations'' preserves (F2)--(F8) (Lemma~\ref{L:filtrationProperties}). 

\begin{remark}
This line of reasoning resembles e.g. a classical proof of the Kripke finite model property of $\slg{S4}$ by showing that the transitive closure of the minimal filtration is a filtration, that the minimal filtration itself is reflexive, and that transitive closure preserves reflexivity (cf.~\cite{CZbook}).
\end{remark}

\begin{definition} 
Set $\kf{M}^+ := \langle X_c/\mathord{\sim}, R, P^+, \fval{\cdot} \rangle$, where:
\begin{itemize}
	\item $\Gamma \sim \Delta$ iff $\Gamma \in \fval{\phi}_c \Eq \Delta \in\fval{\phi}_c$ for all $\phi \in \Sub\phi_0$. Denote the projection $X_c \to X_c/\mathord{\sim}$ by $\proj$;
	\item $R$ is the transitive closure of $\set{\langle \proj(x), \proj(y)\rangle}{x\mathrel{R_c}y}$;
	\item $P := \set{ \langle \proj(\Gamma), \proj[S], \proj(\Delta)\rangle }{ \langle \Gamma,S,\Delta\rangle \in P_c}$;
	\item $P^+$ is the smallest set extending $P$ and satisfying (F1);
	\item $\fval{p} := \proj[\fval{p}_c]$ for all $p \in \PV$.
\end{itemize}
\end{definition}

\begin{remark}\label{R:concat}
Observe that $\langle x,S,y\rangle \in P^+$ iff there exists a sequence $x_0,\dotsc,x_k$ with $k\geq 1$, $x_0=x$, $x_k=y$, and $x_1,\dotsc,x_{k-1} \in S$ such that $\langle x_i, S,x_{i+1}\rangle \in P$ for each $i=0,\dotsc,k-1$.
\end{remark}

\begin{lemma}\label{L:filtrationLemma}
For every $\chi \in \Sub\phi_0$ we have $\fval{\chi}_c = \proj^{-1}[\fval{\chi}]$.
\end{lemma}
\begin{proof}
By induction on $\chi$. The base case and the step for $\chi=\phi\to\psi$ or $\chi=\Boxd\phi$ are as in the standard filtration construction for $\slg{K4}$ (cf. \cite{CZbook}); we only consider $\chi=\gamma(\phi,\psi)$. 
If $\Gamma \in \fval{\gamma(\phi,\psi)}_c$ is witnessed by $\langle \Gamma, \fval{\phi}_c, \Delta\rangle \in P_c$, then $\langle \proj(\Gamma), \fval{\phi}, \proj(\Delta)\rangle \in P \subseteq P^+$ witnesses $\proj(\Gamma) \in \fval{\gamma(\phi,\psi)}$. 

Conversely, suppose $\proj(\Gamma) \in \fval{\gamma(\phi,\psi)}$ is witnessed by $\langle \proj(\Gamma),\fval{\phi}, y\rangle \in P^+$, and fix a corresponding sequence $x_0,\dotsc,x_k$ as in Remark~\ref{R:concat}. 
For each $i<k$, fix also $\Gamma_i \in \proj^{-1}(x_i)$ and $\Gamma'_{i+1} \in \proj^{-1}(x_{i+1})$ with $\langle \Gamma_i, \fval{\phi}_c, \Gamma_{i+1}'\rangle \in P_c$. Let us show by descending induction on $i=k-1,\dotsc,0$ that $\Gamma_i \in \fval{\gamma(\phi,\psi)}_c$. For $i=k-1$, it suffices to note that $\langle \Gamma_{k-1},\fval{\phi}_c, \Gamma_k'\rangle \in P_c$ and $\Gamma_k' \in g^{-1}(y) \subseteq g^{-1}[\fval{\psi}] = \fval{\psi}_c$. Now consider $i<k-1$ and suppose $\Gamma_{i+1} \in \fval{\gamma(\phi,\psi)}_c$. As $\Gamma_{i+1},\Gamma_{i+1}' \in g^{-1}(x_{i+1})$, we have $\Gamma_{i+1}\sim\Gamma_{i+1}'$, thus $\Gamma_{i+1}' \in \fval{\gamma(\phi,\psi)}_c$, hence $\Gamma_i \in \fval{\gamma(\phi,\phi\wedge\gamma(\phi,\psi))}_c$. 
By Lemma~\ref{L:canonicalModelTruth} we have $\kf{M}_c \models \slg{TLR}$, thus  $ \fval{\gamma(\phi,\phi\wedge\gamma(\phi,\psi))}_c \subseteq \fval{\gamma(\phi,\psi)}_c$ by (A1), hence $\Gamma_i \in \fval{\gamma(\phi,\psi)}_c$. We conclude that $\Gamma_0 \in \fval{\gamma(\phi,\psi)}_c$. Since $\Gamma,\Gamma_0 \in g^{-1}(x_0)$, we have $\Gamma \sim\Gamma_0$, hence $\Gamma \in \fval{\gamma(\phi,\psi)}_c$.
\end{proof}

\begin{lemma}\label{L:minimalFiltrationProperties}
The flanked Kripke frame $\langle X_c/\mathord{\sim}, R, P\rangle$ satisfies (F2)--(F8). 
\end{lemma}
\begin{proof}
First note that:
\begin{itemize}
	\item All instances of (A2)--(A8) are valid in $\kf{M}_c$. Indeed, $\kf{M}_c \models \slg{TLR}$ by Lemma~\ref{L:canonicalModelTruth}.
	\item For every $T \subseteq X_c/\mathord{\sim}$ there exists a formula $\phi \in \Lang$ with $\proj^{-1}[T]=\fval{\phi}_c$. Moreover, if $T$ is an $R$-upset, then $\fval{\phi}_c = \fval{\Boxc\phi}_c$. Indeed, an appropriate $\phi$ may be built as a Boolean combination of $\Sub\phi_0$. 
If $T$ is an $R$-upset and $\Gamma \in \fval{\phi}_c$, then for every $\Delta \in R_c(\Gamma)$ we have $g(\Gamma) \mathrel{R} g(\Delta)$ and $g(\Gamma) \in T$, thus $g(\Delta) \in T$, whence $\Delta \in \fval{\phi}_c$, hence $\Gamma \in \fval{\Boxc\phi}_c$. 
\end{itemize}

Now we check each of the requisite properties.

(F2): Consider $x \in X_c/\mathord{\sim}$. Fix some $\Gamma \in \proj^{-1}(x)$ and $\phi$ such that $\fval{\phi}_c = \proj^{-1}(x)$. We have $\Gamma \in \fval{\phi}_c \subseteq \fval{\gamma(\phi,\phi)}_c$ by (A2), whence $\langle \Gamma, \fval{\phi}_c,\Gamma' \rangle \in P_c$ for some $\Gamma' \in \fval{\phi}_c$, hence $\langle x, \{x\},x \rangle \in P$.

For (F3)--(F8), consider $\langle x, S, y \rangle \in P$, and fix $\Gamma \in \proj^{-1}(x)$ and $\Delta \in \proj^{-1}(y)$ such that $\langle \Gamma, \proj^{-1}[S], \Delta\rangle \in P_c$. 

(F3): Note that $\Gamma\not\in\fval{\gamma(\bot,\top)}_c$ by (A3), thus $\proj^{-1}[S] \neq \varnothing$, hence $S \neq \varnothing$.

(F4): Fix $\phi$ and $\psi$ such that $\fval{\phi}_c = \proj^{-1}[S]$ and $\fval{\psi}_c=\proj^{-1}(x)$. We have $\Gamma \in \fval{\psi}_c$, thus $\Gamma \not\in\fval{\gamma(\phi,\neg\gamma(\phi,\psi))}_c$ by (A4), whence $\Delta\not\in\fval{\neg \gamma(\phi,\psi)}_c$, hence $\langle \Delta, \fval{\phi}_c, \Gamma'\rangle \in P_c$ for some $\Gamma' \in \fval{\psi}_c$, thus $\langle y,S,x\rangle \in P$.

(F5): Consider $z \in S$ and fix $\phi, \psi$ with $\fval{\phi}_c = \proj^{-1}[S]$ and $\fval{\psi}_c = \proj^{-1}(y)$. If there exists $\Theta \in \proj^{-1}(z) \cap \fval{\gamma(\phi,\psi)}_c$, then $\langle \Theta, \fval{\phi}_c, \Delta'\rangle \in P_c$ for some $\Delta' \in \fval{\psi}_c$, whence $\langle z, S,y \rangle \in P$. Now suppose $\proj^{-1}(z) \cap \fval{\gamma(\phi,\psi)}_c = \varnothing$.  We have $\Gamma \in\fval{\gamma(\phi,\psi)}_c \subseteq \fval{\gamma(\phi\wedge \gamma(\phi,\psi),\psi)}_c$ by (A5), whence $\langle \Gamma, \proj^{-1}[S] \cap \fval{\gamma(\phi,\psi)}_c, \Delta'\rangle \in P_c$ for some $\Delta' \in \proj^{-1}(y)$, thus $\langle \Gamma, \proj^{-1}[S \setminus \{z\}], \Delta'\rangle \in P_c$, hence $\langle x, S \setminus \{z\}, y \rangle \in P$.

(F6): Consider an $R$-upset $U$ with $x \in U$. Fix $\phi$, $\psi$, and $\chi$ such that $\fval{\phi}_c = \proj^{-1}[S]$, $\fval{\Boxc\psi}_c = \fval{\psi}_c = \proj^{-1}[U]$, and $\fval{\chi}_c = \proj^{-1}(y)$. 
Now $\Gamma \in {\fval{\Boxc\psi\wedge \gamma(\phi,\chi)}_c} \subseteq \fval{\gamma(\phi\wedge\Boxc\psi, (\phi\wedge\neg\Boxc\psi)\vee\chi)}_c$ by (A6), whence $\langle \Gamma, \fval{\phi\wedge \Boxc\psi}_c, \Theta\rangle \in P_c$ for some $\Theta \in \fval{(\phi \wedge \neg\Boxc\psi)\vee \chi}_c$, hence $\langle x, {S \cap U}, \proj(\Theta)\rangle \in P$ and $\proj(\Theta) \in (S \setminus U) \cup \{y\}$. 

(F7): Consider $R$-upsets $U_1$ and $U_2$ with $x,y \in S \subseteq U_1 \cup U_2$. Denote by $T$ the set of all $s \in S$ for which there exists an $\langle S,U_1,U_2,s,y\rangle$-sequence. 
Fix $\phi_1$, $\phi_2$, $\psi$, and $\chi$ such that $\fval{\Boxc\phi_i}_c=\fval{\phi_i}_c= \proj^{-1}[U_i]$, $\fval{\psi}_c = \proj^{-1}[T]$, and $\fval{\chi}_c = \proj^{-1}[S]$. 
It suffices to show that $\langle \Gamma, \proj^{-1}[S], \Delta\rangle \in P_c$ witnesses 
$\Gamma \in \fval{\widehat\gamma(\chi\wedge (\Boxc\phi_1\vee\Boxc\phi_2)\wedge(\widehat\gamma(\chi\wedge\Boxc\phi_1,\psi)\vee \widehat\gamma(\chi\wedge\Boxc\phi_2,\psi)\to\psi), \psi)}_c$, for then $\Gamma \in \fval{\psi}_c$ by (A7) and thus ${x \in T}$.
Since $y \in T$, we have $\Delta \in \fval{\psi}_c$. It is immediate that $\Gamma,\Delta \in g^{-1}[S]= \fval{\chi\wedge (\Boxc\phi_1\vee\Boxc\phi_2)}_c$. Now it remains to verify that $\fval{\widehat\gamma(\chi\wedge\Boxc\phi_i,\psi)\to \psi}_c=X_c$ for $i=1,2$.
Suppose $\Theta \in \widehat\gamma(\chi\wedge\Boxc\phi_i,\psi)$. Then $\Theta \in \fval{\chi\wedge \Boxc \phi_i}_c$ and $\langle \Theta, \fval{\chi\wedge\Boxc\phi_i}_c,\Theta'\rangle \in P_c$ for some $\Theta' \in \fval{\chi\wedge\Boxc\phi_i\wedge \psi}_c$, thus $\proj(\Theta) \in S \cap U_i$ and $\proj(\Theta') \in S \cap U_i \cap T$. Now there exists an $\langle S,U_1,U_2,g(\Theta'), y\rangle$-sequence, and appending $g(\Theta)$ to its beginning we obtain an $\langle S,U_1,U_2,g(\Theta),y\rangle$-sequence, whence $g(\Theta) \in T$, hence $\Theta \in \fval{\psi}_c$.

(F8): Suppose $S=\{y\}$. Fix $\phi$ with $\fval{\phi}_c = \proj^{-1}(y)$. 
If $\fval{\phi}_c = \fval{\phi \wedge \Boxd\neg\phi}_c$, then $\langle \Gamma,\proj^{-1}(y),\Delta\rangle \in P_c$ witnesses $\Gamma \in \fval{\gamma(\phi\wedge\Boxd\neg\phi,\top)}_c$, whence $\Gamma \in \fval{\phi}_c$ by (A8), hence $x=\proj(\Gamma)=y$. Otherwise, fix $\Theta \in {\fval{\phi \wedge \neg \Boxd\neg\phi}_c}$. The set $\{\phi\} \cup \set{\psi}{\Boxd\psi \in \Theta}$ is $\slg{TLR}$-consistent, for if $\vdash \bigwedge \psi_i \to \neg \phi$ then $\vdash \bigwedge \Boxd\psi_i \to \Boxd\neg \phi$. Fix some $\Theta' \in \fval{\phi}_c \cap \bigcap_{\Boxd\psi\in\Theta} \fval{\psi}_c$. Now $\Theta \mathrel{R_c} \Theta'$, whence $y\mathrel{R}y$. 
\end{proof}

\begin{lemma}\label{L:filtrationProperties}
The frame $\langle X_c/\mathord{\sim}, R, P^+\rangle$ is suitable.
\end{lemma}
\begin{proof}
The frame is clearly finite and satisfies (F1). Property (F2) holds for $P^+$ since it holds for $P$ by Lemma~\ref{L:minimalFiltrationProperties} and $P \subseteq P^+$. By Lemma~\ref{L:minimalFiltrationProperties}, $P$ also satisfies (F3)--(F8), hence $P^+$ satisfies (F3)--(F8) for all triples $\langle x,S,y\rangle$ that are in $P$. Now it suffices to check that the set of triples for which $P^+$ satisfies (F3)--(F8) is closed under (F1), i.e., that if $P^+$ satisfies (F3)--(F8) for some $\langle x,S,w\rangle \in P^+$ and $\langle w,S,y\rangle \in P^+$, with $w \in S$, then $P^+$ also satisfies (F3)--(F8) for $\langle x,S,y\rangle$. We do this for each of these properties separately.

(F3): We have $w \in S$, thus $S \neq \varnothing$.

(F4): Since (F4) holds for $\langle x,S,w\rangle$ and $\langle w,S,y\rangle$, we have $\langle y,S,w\rangle, \langle w,S,x\rangle \in P^+$, hence $\langle y,S,x\rangle \in P^+$ by (F1).

(F5): Consider $z \in S$. 
If $z = w$ or $\langle z,S,y\rangle \in P^+$, there is nothing to prove. If $\langle z,S,w\rangle \in P^+$, then $\langle z,S,y\rangle \in P^+$ by (F1) since $\langle w,S,y\rangle \in P^+$. In all other cases, $\langle x,S \setminus \{z\},w\rangle \in P^+$ by (F5) for $\langle x,S,w\rangle$, $\langle w,S\setminus \{z\},y\rangle \in P^+$ by (F5) for $\langle w,S,y\rangle$, and $w \in S \setminus \{z\}$, hence $\langle x, S \setminus \{z\},y\rangle \in P^+$ by (F1).

(F6): Consider an $R$-upset $U$ with $x \in U$. By (F6) applied to $\langle x,S,w\rangle$ we have $\langle x,S\cap U,z\rangle \in P^+$ for some $z \in (S \setminus U) \cup \{w\}$. If $z \in S \setminus U$, there is nothing to prove. Otherwise, $z=w \in U$. Now, (F6) applied to $\langle w,S,y\rangle$ yields $\langle w,S\cap U,z'\rangle \in P^+$ for some $z' \in (S\setminus U) \cup \{y\}$. Since $\langle x,S\cap U,w\rangle \in P^+$ and $w \in U$, we obtain $\langle x,S\cap U,z'\rangle \in P$ by (F1).

(F7): Consider $R$-upsets $U_1$ and $U_2$ with $x,y \in S \subseteq U_1\cup U_2$. 
By (F7) for $\langle x,S,w\rangle$ and $\langle w,S,y\rangle$ there exist an $\langle S,U_1,U_2,x,w\rangle$-sequence and an $\langle S,U_1,U_2,w,y\rangle$-sequence; concatenating them we obtain an $\langle S,U_1,U_2,x,y\rangle$-sequence.

(F8): If $S \neq \{y\}$, there is nothing to prove. If $S=\{y\}$, then $w \in S$ implies $y=w$, hence (F8) for $\langle x,S,w\rangle$ implies (F8) for $\langle x,S,y\rangle$.
\end{proof}

Summarizing this section, we have: 

\begin{proposition}\label{P:flankedFMP}
If $\slg{TLR}\nvdash \phi_0$, then $\phi_0$ is refuted in a suitable frame of size at most $2^{|\Sub\phi_0|}$. 
\end{proposition}
\begin{proof}
The flanked Kripke frame $\langle X_c/\mathord{\sim}, R, P^+\rangle$ is suitable by Lemma~\ref{L:filtrationProperties}. Clearly, $|X_c/\mathord{\sim}| \leq 2^{|\Sub\phi_0|}$. 
If $\slg{TLR}\nvdash \phi_0$, then $\kf{M}_c \not\models \phi_0$ by Lemma~\ref{L:canonicalModelTruth}, whence $\kf{M}^+ \not\models \phi_0$ by Lemma~\ref{L:filtrationLemma}.
\end{proof}

\section{Tree Unraveling}\label{S:unraveling} 

In this section, we show that every formula refuted in a suitable frame is also refuted in a metric space. First, we introduce \emph{path-morphisms} that preserve the validity of $\Lang$-formulas (cf.~\cite{BBCFG24}). Then, given a suitable frame $X$, it will suffice to produce a path-morphism $f: T \to X$ from some metric space $T$.

\begin{definition}
For a topological space $\langle T,\tau\rangle$ and a flanked Kripke frame $\langle X,R,P\rangle$, a map $f: T \to X$ is a \emph{path-morphism} if:
\begin{itemize}[leftmargin=45pt]
	\item[($\Diamond$-forth)] every $a \in T$ has a neighborhood $U \in \tau$ with $f[U \setminus \{a\}] \subseteq R(f(a))$;
	\item[($\Diamond$-back)] if $f(a) R y$ and $a \in U \in \tau$, then there exists $b \in U \setminus \{a\}$ with $f(b)=y$; 
	\item[($\gamma$-forth)] if $\delta$ is a path in $\langle T,\tau\rangle$, then $\langle f(\delta(0)), f[\delta[(0,1)]], f(\delta(1))\rangle \in P$; and
	\item[($\gamma$-back)] if $\langle f(a),S,y\rangle \in P$, then there exists a path $\delta$ in $\langle T,\tau\rangle$ with $\delta(0)=a$, $f[\delta[(0,1)]] \subseteq S$, and $f(\delta(1)) = y$.
\end{itemize}
\end{definition}

\begin{lemma}\label{L:pathMorphism}
Let $f : \langle T,\tau\rangle \to \langle X,R,P\rangle$ be a path-morphism and $\fval{\cdot} : X \to \PV$. Put $\fval{p}_T := f^{-1}[\fval{p}]$ for all $p \in \PV$. Then $a \in \fval{\chi}_T \Eq f(a) \in \fval{\chi}$ for all $\chi \in \Lang$ and $a \in T$.
\end{lemma}
\begin{proof}
By induction on $\chi$. The step for $\chi=\phi\to\psi$ is trivial. 

Consider $\chi=\Boxd\phi$. If $f(a) \in \fval{\Boxd\phi}$, then $R(f(a)) \subseteq \fval{\phi}$, whence $f^{-1}[R(f(a))] \subseteq\fval{\phi}_T$ by the induction hypothesis for $\phi$, hence the set $U$ given by ($\Diamond$-forth) witnesses $a \in \fval{\Box\phi}_T$. 
Conversely, suppose $a\in \fval{\Boxd\phi}_T$, i.e., there exists $U \in \tau$ with $a \in U$ and $U \setminus \{a\} \subseteq \fval{\phi}_T$. By ($\Diamond$-back), for every $y \in R(f(a))$ there exists $b \in U \setminus \{a\} \subseteq \fval{\phi}_T$ with $f(b)=y$, whence $y \in \fval{\phi}$ by the induction hypothesis for $\phi$. Therefore, $R(f(a)) \subseteq \fval{\phi}$, i.e., $f(a) \in \fval{\Boxd\phi}$. 

Consider $\chi=\gamma(\phi,\psi)$. If $a \in \fval{\gamma(\phi,\psi)}_T$ is witnessed by a path $\delta$, then $\langle f(a), f[\delta[(0,1)]], f(\delta(1))\rangle \in P$ by ($\gamma$-forth), hence $f(a) \in \fval{\gamma(\phi,\psi)}$. Conversely, if $f(a) \in \fval{\gamma(\phi,\psi)}$ is witnessed by $\langle f(a),\fval{\phi},y\rangle \in P$, then by \mbox{($\gamma$-back)} there exists a path $\delta$ in $\langle T,\tau\rangle$ witnessing $a \in \fval{\gamma(\phi,\psi)}_T$.
\end{proof}

\begin{example}\label{E:cantor}
Take $\langle X,R,P\rangle$ from Example~\ref{E:positive}. Consider the Cantor set $\mathcal{C} \subseteq [0,1]$ and present $[0,1] \setminus \mathcal{C}$ as a union of pairwise-disjoint intervals $\bigcup_{i \in \mathbb{Q}} (\alpha_i,\beta_i)$. Now consider $\pi:[0,1] \to X$ given by $\pi[[0,1] \setminus \bigcup_{i \in \mathbb{Q}} [\alpha_i,\beta_i]] = \{c\}$ and $\pi \restriction [\alpha_i,\beta_i] = \pi'$ for every $i$, where $\pi'(0)=\pi'(1) := b$ and $\pi'[(0,1)] := \{a\}$. A direct inspection shows that $\pi$ is a path-morphism with respect to the usual topology on $[0,1]$. By Lemma~\ref{L:pathMorphism}, it follows that all $\Lang$-formulas refuted in $\langle X,R,P\rangle$ are also refuted in $[0,1]$.
\end{example}

In general, given a suitable frame $X$, we will construct a metric tree-like structure $T$ whose ``edges'' are isomorphic copies of $[0,1]$, and a path-morphism $f: T \to X$ will be defined in terms of its restrictions to individual ``edges.'' For $f$ to satisfy ($\Diamond$-forth) and ($\gamma$-forth), it is necessary that said restrictions themselves satisfy ($\Diamond$-forth) and ($\gamma$-forth), i.e., that they are maps of the following form:

\begin{definition}
A \emph{good path} in a suitable frame $\langle X,R,P\rangle$ is a map $[0,1]\to X$ satisfying ($\Diamond$-forth) and:
\begin{itemize}
	\item[(G)] for all $0\leq t<u\leq 1$ we have $\langle \pi(t), \pi[(t,u)], \pi(u) \rangle \in P$.
\end{itemize}
\end{definition}

\begin{remark}
It easily follows from (F1) and (F4) that (G) is equivalent to ($\gamma$-forth). A map $[0,1]\to X$ satisfies ($\Diamd$-forth) iff preimages of $R$-upsets are open and preimages of irreflexive singletons are discrete. In particular, $\iota_{x,x}$ is a good path iff $x \mathrel{R} x$.
\end{remark}

\begin{definition}
A map $\pi: [0,1] \to X$ \emph{matches} a triple $\langle x, S, y\rangle \in X \times \cl{P}(X) \times X$ if $\pi(0)=x$, $\pi[(0,1)] \subseteq S$, and $\pi(1)=y$.
\end{definition}

We claim that each ``non-constant'' triple in $P$ is matched either by a good path or by a constant (Lemma~\ref{L:matchingGoodPaths}). This will provide us a sufficient supply of good paths to define a path-morphism $f: T \to X$ (Proposition~\ref{P:networkTree}).

\begin{example}
Each triple in $P$ from Examples~\ref{E:positive} and \ref{E:cantor} is matched by one of the following good paths: $\iota_{a,a}$, $\iota_{b,b}$, $\iota_{c,c}$, $\iota_{b,a}$, $\iota_{b,a}\restriction [1,0]$, $\pi'$, the map given by $\delta(0):=c$ and $\delta(t):=\pi(t)$ for $t>0$, the map given by $\delta'(1):=c$ and $\delta'(t):=\pi(t)$ for $t<1$, the concatenation of $\iota_{a,b}$ with $\delta'$, or the concatenation of $\delta$ with $\iota_{a,b}\restriction [1,0]$. 
\end{example}

We will construct the requisite good paths from ``simpler'' good paths using several specific operations (among them are concatenation, reversion, the way $\pi$ is constructed from $\pi'$ in Example~\ref{E:cantor}, etc.). 
The following lemma, informally speaking, states that all elements of $P$ can be ``constructed'' from ``simpler'' elements of $P$ using the corresponding set of operations on triples.

\begin{lemma}\label{L:triplesOrdering}
Let $\langle X, R, P \rangle$ be a suitable frame. Then there exists a strict partial order $\vartriangleleft$ on $P$ such that for every $\langle x, S,y\rangle \in P$ one of the following holds:
\begin{itemize}
	\item[(C1)] $x=y \in S$; 
	\item[(C2)] $\langle y,S,x\rangle \vartriangleleft \langle x,S,y\rangle$;
	\item[(C3)] there is a sequence $x_0,\dotsc,x_k$ with $k \geq 1$, $x=x_0$, $x_k=y$, and $x_1,\dotsc,x_{k-1} \in S$, such that for each $i<k$ there exists $S_i' \subseteq S$ with $P \owns \langle x_i, S_i',x_{i+1}\rangle \vartriangleleft \langle x,S,y\rangle$;
	\item[(C4)] we have:
	\begin{itemize}
		\item[(C4a)] $x \neq y \in S \subseteq R(x)$;
		\item[(C4b)] $P \owns \langle z,S,w\rangle \vartriangleleft \langle x,S,y\rangle$ for all $z,w \in S$; and 
		\item[(C4c)] $\langle x,S,z\rangle \in P$ for all $z \in S$;
	\end{itemize}
	\item[(C5)] there exists $U \subseteq X$ such that:
	\begin{itemize}
	\item[(C5a)] $x,y\in S \setminus U$;	
	\item[(C5b)] $(S \setminus U) \times S \subseteq R$;
	\item[(C5c)] for every $z \in S\cap U$ there exists $w_z \in S \setminus U$ with $P \owns \langle z, S \cap U,w_z\rangle \vartriangleleft \langle x,S,y\rangle$; and
	\item[(C5d)] $\langle z,S,w\rangle \in P$ for all $z,w \in S$.
	\end{itemize}
\end{itemize}
\end{lemma}
\begin{proof}
For $z \in X$, put $\refcl{R}(z) := \{z\} \cup R(z)$; note that $\refcl{R}(z)$ is an $R$-upset. 
Put $\langle x,S,y\rangle \vartriangleleft \langle x',S',y'\rangle$ iff one of the following holds:
\begin{itemize}
	\item[(L1)] $S\subsetneq S'$;
	\item[(L2)] $S=S'$ and $y \in S \not \owns y'$;
	\item[(L3)] $S=S'$, $x \in S \not\owns x'$, and either $y,y' \in S$ or $y,y'\not\in S$;
	\item[(L4)] $x,y,x',y'\in S=S'$ and $\refcl{R}(y) \supsetneq \refcl{R}(y')$; 
	\item[(L5)] $x,y,x',y' \in S=S'$, $\refcl{R}(y)=\refcl{R}(y')$, and $\refcl{R}(x) \supsetneq \refcl{R}(x')$.
\end{itemize}
Observe that $\vartriangleleft$ is a strict partial order. Now consider $\langle x,S,y\rangle \in P$ and suppose that (C1), (C2), and (C3) do not hold. Then:
\begin{itemize}
	\item[(i)] For all $S' \subsetneq S$ we have $\langle x, S' , y\rangle \not\in P$. Indeed, otherwise (C3) is true with $k=1$.
	\item[(ii)] For every $z \in S$ we have $\langle x, S, z\rangle \in P$ and $\langle z,S,y\rangle \in P$.
Indeed, $\langle x,S\setminus \{z\}, y\rangle \not \in P$ by (i); it remains to apply (F5) together with (F4).  
	\item[(iii)] If $y \in S$, then $\langle z,S,w\rangle \in P$ for all $z,w \in S$. Indeed, $\langle z,S,y\rangle, \langle y, S, w\rangle \in P$ by (ii) and (F4), hence $\langle z,S,w\rangle \in P$ by (F1). 
	\item[(iv)] If an $R$-upset $U$ satisfies $S \not\subseteq U$ and $x \in U$, then there exists $z \in S\setminus U$ with $\langle z,S,y\rangle \ntriangleleft \langle x,S,y\rangle$. 
Indeed, by (F6) we have $\langle x, S \cap U, z \rangle \in P$ for some $z \in (S \setminus U) \cup \{y\}$. 
By (i), $z\neq y$. Now we have $z \in S$, $P \owns \langle x,S \cap U, z\rangle \vartriangleleft \langle x,S,y\rangle$ by (L1), and $\langle z,S,y\rangle \in P$ by (ii). Therefore, $\langle z,S,y\rangle \ntriangleleft \langle x,S,y\rangle$, for otherwise (C3) holds with $k=2$ and $x_1=z$.
\end{itemize}

Now, to show that either (C4) or (C5) holds for $\langle x,S,y\rangle$, consider two cases.

\NewCase\textbf{Case 1: $x \not \in S$.} Let us establish (C4). 
By (F3) we have $S \neq \varnothing$; fix some $z' \in S$. If $y \not \in S$, we have $P \owns \langle x,S,z'\rangle \vartriangleleft \langle x,S,y\rangle$ by (ii) and (L2), and $P \owns \langle z',S,y\rangle \vartriangleleft \langle x,S,y\rangle$ by (ii) and (L3), hence (C3) holds with $k=2$. Therefore, $y \in S$.

By (L3) we have $\langle z,S,y\rangle \vartriangleleft \langle x,S,y\rangle$ for all $z \in S$, thus by (iv) applied to $U:=\refcl{R}(x)$ we obtain $S\subseteq \refcl{R}(x)$, hence $S \subseteq R(x)$, i.e., (C4a) is true. Condition (C4b) follows from (iii) and (L3), and (C4c) holds by~(ii).

\NewCase\textbf{Case 2: $x\in S$.} 
If $y \not \in S$, then $\langle y,S,x\rangle \vartriangleleft\langle x,S,y\rangle$ by (L2), whence (C2) holds. Therefore, $y \in S$.

Fix an $R$-upset $U$ which is maximal with respect to $U \subsetneq \refcl{R}[S]$. (Note that $U$ may be empty.) Since $\refcl{R}[S] \not\subseteq U$, we have $S \not\subseteq U$. Let us show that (C5) is true for $U$.

First we claim that every $z \in S \setminus U$ satisfies $\refcl{R}(z)=\refcl{R}[S]$. Indeed, by maximality of $U$ we have $U \cup \refcl{R}(z) = \refcl{R}[S]$. 
Using (F7), fix an $\langle S, U,\refcl{R}(z),x,y\rangle$-sequence $x_0,\dotsc,x_k$. 
Since (C1) does not hold, we have $x \neq y$ and thus $k>0$.
Since (C3) does not hold, there exists $i<k$ with $\langle x_i, S \cap U, x_{i+1}\rangle \ntriangleleft \langle x,S,y\rangle$ or $\langle x_i, S \cap \refcl{R}(z), x_{i+1}\rangle \ntriangleleft \langle x,S,y\rangle$. 
As (L1) does not apply, it follows that $S \cap U=S$ or $S \cap \refcl{R}(z) =S$. Since $S \not\subseteq U$, we conclude that $S \subseteq \refcl{R}(z)$ and thus $\refcl{R}(z) = \refcl{R}[S]$.

To verify (C5a), first assume $x \in U$. Since $S\not\subseteq U$, by (iv) there exists $z \in S \setminus U$ with $\langle z,S,y\rangle \ntriangleleft \langle x,S,y\rangle$. But $x,y,z \in S$ and $\refcl{R}(x) \subseteq U \subsetneq \refcl{R}[S] = \refcl{R}(z)$, thus $\langle z,S,y\rangle \vartriangleleft \langle x,S,y\rangle$ by (L5), which is a contradiction. Now assume $x \in S \setminus U$ and $y \in U$. Then $\refcl{R}(y) \subsetneq \refcl{R}[S] =\refcl{R}(x)$, thus $\langle y,S,x\rangle \vartriangleleft \langle x,S,y\rangle$ by (L4), hence (C2) holds. Therefore, $x,y \in S \setminus U$.

To verify (C5b), recall that $S \subseteq \refcl{R}[S] = \refcl{R}(z)$ for all $z \in S \setminus U$. It remains to note that $\refcl{R}$ is total on $S \setminus U$ and $x,y \in S \setminus U$ are distinct, thus all $z \in S \setminus U$ are $R$-reflexive.

To check (C5c), consider $z \in S \cap U$. By (ii) we have $\langle z, S,y\rangle \in P$, thus by (F6) there exists $w_z \in S \setminus U \cup \{y\} = S \setminus U$ satisfying $\langle z, S \cap U, w_z \rangle \in P$.

Condition (C5d) follows from (iii).
\end{proof}

\begin{lemma}\label{L:matchingGoodPaths}
Let $\langle X, R, P \rangle$ be a suitable frame. Then, for every triple $\langle x,S,y\rangle \in P$, either $x=y \in S$ or there exists a good path matching $\langle x,S,y\rangle$.
\end{lemma}
\begin{proof}
Fix a strict partial order $\vartriangleleft$ as in Lemma~\ref{L:triplesOrdering}. 
Since $X$ and $P$ are finite, this $\vartriangleleft$ is well-founded. The proof proceeds by $\vartriangleleft$-induction on $P$.

\NewCase\textbf{Case (C1).} Since $x=y \in S$, there is nothing to prove.

\NewCase\textbf{Case (C2).} Since $\langle y,S,x\rangle \vartriangleleft \langle x,S,y\rangle$, we have $x\neq y$. By the induction hypothesis, some good path $\pi$ matches $\langle y,S,x\rangle$. Now, $\pi \restriction [1,0]$ matches $\langle x,S,y\rangle$; it is a good path by (F4).

\NewCase\textbf{Case (C3).} Consider a requisite sequence $x_0,\dotsc,x_k$.
Without loss of generality, either $x_0\neq x_1 \neq \dotso \neq x_k$ or $k=1$, for we can omit duplicate adjacent $x_i$'s. 
If $k=1$ and $x=y \in S_0'$, then $x=y\in S$. Otherwise, for each $i<k$ the induction hypothesis yields a good path $\pi_i$ matching $\langle x_i,S_i',x_{i+1}\rangle$. Now the map $\pi:[0,1]\to X$ given by $\pi \restriction [\frac{i}{n}, \frac{i+1}n] = \pi_i$ matches $\langle x,S,y\rangle$. Since each $\pi_i$ satisfies ($\Diamond$-forth), $\pi$ also satisfies ($\Diamond$-forth). Since each $\pi_i$ satisfies (G), $\pi$ also satisfies (G) by (F1).

\NewCase\textbf{Case (C4) with $S=\{y\}$.}
The map $\iota_{x,y}$ matches $\langle x,S,y\rangle$. We have $x \mathrel{R} y$ by (C4a) and $y \mathrel{R} y$ by (F8), hence $\iota_{x,y}$ satisfies ($\Diamond$-forth). By (F2) we have $\langle y,\{y\},y\rangle \in P$, thus $\iota_{x,y}$ satisfies (G).

\NewCase\textbf{Case (C4) with $\{y\} \subsetneq S$.} 
Denote by $z_0,\dotsc,z_{k-1}$ an enumeration of $S$ such that $z_0=y$, and put $z_k = z_0$. For each $i<k$ we have $P \owns \langle z_{i+1}, S,z_{i}\rangle \vartriangleleft \langle x,S,y\rangle$ by (C4b) and $z_{i+1} \neq z_i$, thus the induction hypothesis yields a good path $\pi_i$ matching $\langle z_{i+1}, S, z_{i}\rangle$. 
Consider $\pi:[0,1]\to X$ given by $\pi(0)=x$ and $\pi \restriction [\frac1{i+2}, \frac1{i+1}] = \pi_{i \bmod k}$ for all $i \geq 0$. 
Clearly, $\pi$ matches $\langle x, S, y\rangle$. 
This $\pi$ satisfies ($\Diamd$-forth) at every $t \in (0,1]$ since $\pi_i$'s satisfy \mbox{($\Diamd$-forth)}, and at $t=0$ since $\pi[(0,1)]=S \subseteq R(x)$ by (C4a).
For all $u \in (0,1]$ we have $\langle \pi(0), \pi[(0,u)],\pi(u)\rangle = \langle x,S,\pi(u) \rangle \in P$ by (C4c). Condition (G) for $0<t<u\leq 1$ follows easily by (F1) from the fact that $\pi_i$'s satisfy (G).

\NewCase\textbf{Case (C5).} 
For every $z \in S\cap U$, we claim that there exists a good path $\pi_z$ with $z \in \pi_z[(0,1)] \subseteq S$ and $\pi_z(0)=\pi_z(1) \in S \setminus U$. 
Indeed, by (C5c) and the induction hypothesis, there exists a good path $\pi_z'$ matching $\langle z, S\cap U, w_z\rangle$ for some $w_z \in S \setminus U$. Take $\pi_z$ with $\pi_z \restriction [0,\frac12] = \pi_z' \restriction [1,0]$ and $\pi_z \restriction [\frac12,1] = \pi_z'$.

As in Example~\ref{E:cantor}, denote by $\mathcal{C} \subseteq [0,1]$ the standard Cantor set, and present $[0,1] \setminus \cl{C}$ as a union of pairwise-disjoint intervals $\bigcup_{i \in \mathbb{Q}} (\alpha_i,\beta_i)$, with $i<j \Eq \alpha_i<\alpha_j$. Fix $f: \mathbb{Q} \to S$ such that $f^{-1}(z)$ is dense in $\mathbb{Q}$ for every $z \in S$. 
Define $\pi:[0,1]\to X$ as follows:
\begin{itemize}
	\item $\pi \restriction [\alpha_i,\beta_i] = \pi_{f(i)}$ if $f(i) \in S\cap U$ and  $\pi\restriction [\alpha_i,\beta_i]=\iota_{f(i),f(i)}$ if $f(i) \in S \setminus U$;
	\item $\pi(t)=x$ for all $t \in [0,1) \setminus \bigcup_{i \in \mathbb{Q}} [\alpha_i,\beta_i]$, and $\pi(1)=y$.
\end{itemize}

Clearly, $\pi$ matches $\langle x, S, y\rangle$.
Consider $t \in [0,1]$. If $\pi(t)\in S\setminus U$, then $R(\pi(t))\supseteq S=\pi[[0,1]]$ by (C5a) and (C5b). Otherwise, some $i \in\mathbb{Q}$ satisfies $t \in (\alpha_i,\beta_i)$ and $\pi \restriction [\alpha_i,\beta_i]$ satisfies ($\Diamond$-forth). Therefore, $\pi$ satisfies ($\Diamond$-forth).
Towards (G), consider $0\leq t<u\leq 1$. If $[t,u] \subseteq [\alpha_i,\beta_i]$ for some $i$ with $f(i) \in S\cap U$, it suffices to apply (G) for $\pi_{f(i)}$. If $[t,u]\subseteq [\alpha_i,\beta_i]$ with $f(i) \in S\setminus U$, apply (F2). Finally, suppose $[t,u]$ is not contained in any of the intervals $[\alpha_i,\beta_i]$. Then $[t,u]$ intersects at least two of those intervals, say $[\alpha_i,\beta_i]$ and $[\alpha_j,\beta_j]$ with $i<j$. But now for every $z \in S$ there exists $k \in f^{-1}(z) \cap (i,j)$, thus $z \in \pi [ [\alpha_k,\beta_k] ] \subseteq \pi[(t,u)]$. Therefore, $\pi[(t,u)]=S$, hence $\langle \pi(t),\pi[(t,u)],\pi(u)\rangle=\langle \pi(t), S, \pi(u)\rangle \in P$ by (C5d).
\end{proof}

Now, given a suitable frame $\langle X,R, P\rangle$, we employ Lemma~\ref{L:matchingGoodPaths} to produce a metric ``tree'' $T$ and a path-morphism $f:T \to X$. Let us first outline our construction. To assure ($\Diamd$-back) and ($\gamma$-back), it suffices to define $T$ as the closure of $\{a_0\}$ with $f(a_0) := x_0$ under the following two operations:
\begin{itemize}[leftmargin=35pt]
	\item[($\Diamd$-lift)] if $a \in T$ and $f(a) \mathrel{R} y$, isometrically glue $\{0\} \cup \set{\frac1n}{n\geq 1}$ to $T$ with $a \sim 0$, taking $f(\frac1n):=y$ for all $n\geq 1$;
	\item[($\gamma$-lift)] if $a \in T$ and $\langle f(a),S,y\rangle \in P$, isometrically glue a real interval $[0,1]$ to $T$ with $a \sim 0$, taking $f \restriction [0,1] := \pi$, where $\pi:[0,1]\to X$ is a good path matching $\langle f(a), S,y\rangle$. We need not do this operation if $f(a) = y \in S$ (because in this case ($\gamma$-forth) is already witnessed by $\delta := \iota_{a,a}$); otherwise a requisite good path exists by Lemma~\ref{L:matchingGoodPaths}.
\end{itemize}

For this construction, we will show that properties ($\Diamond$-forth) and ($\gamma$-forth) transfer from individual good paths to the whole map $f$.

\begin{remark}
The metric space $T$ we construct is similar to a \emph{real tree} (as defined e.g. in~\cite{Janson}), except that in addition to ``edges'' isomorphic to $[0,1]$ we also have those isomorphic to $\{0\} \cup \set{\frac1n}{n \geq 1}$. 
We define the metric in terms of a height function $h$ and a meet-semilattice $\wedge$ of ``least common ancestors'' in our ``tree''; this is reminiscent of~\cite[Example 10.9]{Janson}. 
In a real tree, every path passes through all the points on the unique arc between its ends (cf. e.g.~\cite{Janson}); we essentially make a similar observation on our metric to verify ($\gamma$-forth).
\end{remark}

\begin{proposition}\label{P:networkTree}
If $\langle X,R,P\rangle$ is a suitable frame and $x_0 \in X$, then there exists a metrizable topological space $\langle T,\tau\rangle$ and a path-morphism $f: T \to X$ with $x_0 \in f[T]$.
\end{proposition}
\begin{proof}
In this proof, $a {}^\smallfrown b$ denotes the concatenation of tuples $a$ and $b$.

Using Lemma~\ref{L:matchingGoodPaths}, fix a set $\Gamma$ containing one good path matching each triple $\langle x,S,y\rangle \in P$, except those with $x=y \in S$. Denote by $T$ the set consisting of the empty tuple $\varnothing$ and all tuples $\langle \pi_1,u_1,\dotsc,\pi_k,u_k\rangle$ such that:
\begin{itemize}
	\item for each $i=1,\dotsc,k$, either:
	\begin{itemize}
		\item $\pi_i$ is an element of $\Gamma$ and $u_i \in (0,1]$; or
		\item $\pi_i = \iota_{x,y}$ for some $x,y \in X$ with $x \mathrel{R} y$, and $u_i = \frac1n$ for some integer $n>0$; and
\end{itemize}
	\item $\pi_1(0)=x_0$ and $\pi_{i}(0)=\pi_{i-1}(u_{i-1})$ for each $i=2,\dotsc,k$.
\end{itemize}
(Note that $\iota_{x,y}$ may or may not be in $\Gamma$. If $\pi_i =\iota_{x,y} \in \Gamma$, we allow any $u_i \in (0,1]$.)

For $a, b \in T$, put $a \prec b$ iff $b$ has a prefix $c {}^\smallfrown \langle \pi,u\rangle$ such that either $a =c$ or $a=c {}^\smallfrown \langle \pi,u'\rangle$ for some $u'< u$. Observe that $\prec$ is a partial order on $T$.

We claim that $\langle T,\prec\rangle$ is a meet-semilattice (i.e., that any two elements of the ``tree'' $\langle T,\prec\rangle$ have a least common ancestor). Indeed, consider $\prec$-incomparable $a,b \in T$. There are unique prefixes $c {}^\smallfrown \langle \pi_a,u_a\rangle$ and $c {}^\smallfrown\langle \pi_b,u_b\rangle$ of $a$ and $b$, respectively, with $\langle \pi_a,u_a\rangle\neq\langle\pi_b,u_b\rangle$. 
If $\pi_a=\pi_b$, then $c {}^\smallfrown \langle \pi_a, \min\{u_a,u_b\}\rangle$ is the meet of $a$ and $b$. Otherwise, $c$ is the meet of $a$ and $b$. We denote $\prec$-meets by $\wedge$.

Put $h(\langle \pi_1,u_1,\dotsc,\pi_k,u_k\rangle):=u_1+\dotsc+u_k$, $h(\varnothing):=0$, and $d(a,b) := h(a)+h(b)-2h(a\wedge b)$ for $a,b \in T$. 
We claim that $d$ is a metric on $T$. Indeed, $d$ is clearly symmetric. 
Since $h$ is $\prec$-monotone, $d$ is positive.
For all $a,b,c\in T$ we have $h(b) \geq \max\{h(a\wedge b),h(b\wedge c)\}$ and $h(a\wedge c) \geq \min\{h(a\wedge b), h(b\wedge c)\}$, thus $h(b)+h(a\wedge c)\geq h(a\wedge b)+h(b\wedge c)$, hence $d(a,b)+d(b,c)\geq d(a,c)$.

Denote by $\tau$ the topology on $T$ induced by $d$. 
Put $f(a^\smallfrown \langle \pi,u\rangle) := \pi(u)$ and $f(\varnothing) := x_0$.
Now we verify that $f$ is a path-morphism from $\langle T,\tau\rangle$ to $\langle X,R,P\rangle$.

\NewCase\textbf{Condition ($\Diamond$-back).} Consider $a \in T$, $y \in R(f(a))$, and $U \in \tau$ with $a\in U$. For every $n>0$, we have $b_{n} := a {}^\smallfrown \langle \iota_{f(a),y}, \frac1n\rangle \in T$ and $d(a,b_{n}) = \frac1n$. Therefore, $b_{n} \in U \setminus \{a\}$ for a sufficiently big $n$.

\NewCase\textbf{Condition ($\gamma$-back).}
Consider $a \in T$ and $\langle f(a), S,y\rangle \in P$.
If $f(a)=y \in S$, it suffices to take $\delta := \iota_{a,a}$. 
Otherwise, fix $\pi \in \Gamma$ that matches $\langle f(a),S,y\rangle$, put $\delta(0) := a$ and $\delta(u):=a{}^\smallfrown \langle\pi,u\rangle$ for $u\in (0,1]$, and observe that $\delta$ is isometric and thus continuous.

\NewCase\textbf{Condition ($\Diamond$-forth).} Consider $a \in T$.
If $a \neq \varnothing$, present $a$ as $a=a' {}^\smallfrown \langle \pi_a,u_a\rangle$. Since $\Gamma$ is finite and all its elements satisfy ($\Diamond$-forth), we can fix a sufficiently small $\eps>0$ such that: 
\begin{itemize}
	\item for every $\pi \in \Gamma$, we have $\pi[(0,\eps)] \subseteq R(f(\pi(0)))$;
	\item if $a\neq\varnothing$ and $\pi_a \in \Gamma$, then $\eps<u_a$ and $\pi_a[(u_a-\eps,\min\{u_a+\eps,1\}) \setminus \{u_a\}] \subseteq R(f(a))$;
	\item if $a \neq \varnothing$ and $\pi_a \not\in \Gamma$, then $\eps<u_a$ and $(u_a - \eps, u_a) \cup (u_a,u_a+\eps)$ is disjoint from $\set{\frac1n}{n\geq 1}$.
\end{itemize}
Consider $b \in T$ with $d(a,b)<\eps$; it suffices to show that $b=a$ or $f(a) \mathrel{R} f(b)$. As $d(a,{a\wedge b}) \leq d(a,b) < \eps$, we have either $a=\varnothing$ or $a\wedge b = a' {}^\smallfrown \langle \pi_a,u\rangle$ for some $u \in (u_a-\eps, u_a]$. 
Present $b$ as $c {}^\smallfrown \langle \pi_1,u_1,\dotsc,\pi_k,u_k\rangle$ with $k\geq 0$, where $c = \varnothing$ if $a=\varnothing$ and $c = a' {}^\smallfrown \langle \pi_a,u'\rangle$ for some $u'\geq u$ otherwise.
In the latter case, $0\leq u'-u \leq d(a \wedge b,b)<\eps$, whence $u' \in (u_a-\eps,u_a+\eps)$, hence either $c=a$ (if $u'=u_a$) or $f(a) \mathrel{R} f(c)$ (otherwise). 
If $k=0$, there is now nothing to prove. Otherwise, for each $i<k$ with $\pi_i \in \Gamma$ we have $u_i \leq d(a\wedge b,b)<\eps$, whence $\pi_i(0) \mathrel{R} \pi_i(u_i)$ by definition of $\eps$. For each $i<k$ with $\pi_i \not\in \Gamma$, we have $\pi_i=\iota_{x,y}$ for some $x \mathrel{R} y$, thus also $\pi_i(0) \mathrel{R} \pi_i(u_i)$. By transitivity of $R$, it follows that $f(c) \mathrel{R} f(b)$. Since $c=a$ or $f(a) \mathrel{R} f(c)$, we conclude that $f(a) \mathrel{R} f(b)$.

\NewCase\textbf{Condition ($\gamma$-forth).} Consider a path $\delta$ in $\langle T,\tau\rangle$.
First observe the following for every $a^* \in T$:
\begin{itemize}
	\item[(i)] the set $T \setminus \{a^*\}$ is the disjoint union of the following open sets:
\begin{itemize}
	\item $\set{b \in T }{ a^* \not\preceq b}$;
	\item $\set{b \in T }{ \text{$a^* {}^\smallfrown \langle \pi\rangle$ is a prefix of $b$}}$, for some $\pi$;
	\item $\set{b \in T} { \text{$a^* = a' {}^\smallfrown \langle \pi,u'\rangle$ and $a' {}^\smallfrown \langle \pi,u\rangle$ is a prefix of $b$ for some $u>u'$}}$;
\end{itemize}
thus, by connectedness of $[0,1]$, either $a^* \in \delta[[0,1]]$ or $\delta[[0,1]]$ is a subset of one of these sets.
	\item[(ii)] if $a^* = a' {}^\smallfrown \langle \pi,u'\rangle$ with $\pi \not\in \Gamma$, then $\set{b \in T}{a^* \preceq b}$ is open. Hence $\delta[[0,1]]$ is a subset either of said set or of its complement, by connectedness of $[0,1]$.
\end{itemize}

Now, to establish ($\gamma$-forth) for $\delta$, consider 5 cases:
\begin{itemize}
\item\emph{Case 1A: $\delta(1) = a{}^\smallfrown \langle \pi,u\rangle$ and either $\delta(0)=a{}^\smallfrown \langle \pi,u'\rangle$ with $u'<u$ or $\delta(0)=a$.} 
If $\delta(0)=a$, put $u' := 0$. Since $\delta(0) \prec \delta(1)$, by (ii) applied to $a^* := \delta(1)$ we have $\pi \in \Gamma$, thus by (G) $\langle \pi(u'),\pi[(u',u)],\pi(u)\rangle \in P$. For every $u'' \in (u',u)$, we have $\delta(0) \prec a {}^\smallfrown \langle \pi,u''\rangle \prec\delta(1)$, hence $a {}^\smallfrown \langle \pi,u''\rangle \in \delta[(0,1)]$ by (i), whence $\pi(u'') = f(a {}^\smallfrown \langle \pi,u''\rangle) \in f[\delta[(0,1)]]$. Therefore, $\pi[(u',u)] \subseteq f[\delta[(0,1)]]$. As $P$ is monotone, we obtain $\langle f(\delta(0)), f[\delta[(0,1)]], f(\delta(1))\rangle \in P$.

\item\emph{Case 1: $\delta(0) \prec\delta(1)$.} 
By induction on the length of $\delta(1)$. Present $\delta(1)$ as $a {}^\smallfrown \langle \pi,u\rangle$. If Case 1A does not apply, we have $\delta(0)\prec a$. By (i) for $a^* := a$, we obtain $a =\delta(t)$ for some $t\in (0,1)$. It remains to apply the induction hypothesis to $\delta \restriction [0,t]$ and Case 1A to $\delta \restriction [t,1]$, and use (F1).

\item\emph{Case 2: $\delta(1) \prec\delta(0)$.} Apply Case 1 to $\delta \restriction [1,0]$ and use (F4).

\item\emph{Case 3: $\delta(0) \not\prec\delta(1)$, $\delta(1) \not\prec\delta(0)$, and $\delta(0)\neq\delta(1)$.} 
Applying (i) to $a^* := \delta(0) \wedge \delta(1)$, we obtain $\delta(t)=\delta(0)\wedge \delta(1)$ for some $t \in (0,1)$. 
Now apply Case 2 to $\delta \restriction [0,t]$ and Case 1 to $\delta\restriction [t,1]$, and use (F1).

\item\emph{Case 4: $\delta(0) = \delta(1)$.} If $\delta$ is constant, it suffices to use (F2). 
Otherwise, fix some $t \in (0,1)$ with $\delta(t) \neq\delta(0)$, apply previous cases to $\delta \restriction [0,t]$ and $\delta \restriction [t,1]$, and use (F1).
\end{itemize}
\end{proof}

Summarizing this and the preceding sections, we have the following:

\begin{theorem}\label{T:decidability}
For a formula $\phi \in \Lang$, the following are equivalent:
\begin{itemize}	
	\item[(1)] $\phi$ is valid over the class of $T_1$ topological spaces;
	\item[(2)] $\phi$ is valid over the class of metric spaces;
	\item[(3)] $\phi$ is valid over the class of suitable frames of size at most $2^{|\Sub\phi|}$; 
	\item[(4)] $\slg{TLR} \vdash \phi$.
\end{itemize}
The set of formulas with these properties is decidable.
\end{theorem}
\begin{proof}
(1) implies (2) since all metric spaces are $T_1$; (2) implies (3) by Lemma~\ref{L:pathMorphism} and Proposition~\ref{P:networkTree}; (3) implies (4) by Proposition~\ref{P:flankedFMP}; (4) implies (1) by Proposition~\ref{P:soundness}. Decidability follows from (3).
\end{proof}

\section{C-semantics}\label{S:interiorModality}

In this section, we axiomatize the c-semantical fragment of $\slg{TLR}$ and observe that it is sound for the class of all topologies. Denote by $\Langc \subseteq \Lang$ the set of formulas that may be written in terms of $\bot$, $\to$, $\gamma$, and $\Boxc$, without explicit occurrences of $\Boxd$. 

\begin{definition}
Denote by $\slg{TLR_c}$ the $\Langc$-calculus consisting of:
\begin{itemize}
	\item all axioms and rules of $\slg{S4}$ for $\Boxc$;
	\item axioms (A0)--(A7) and rule (Mon).
\end{itemize}
\end{definition}

\begin{lemma}\label{L:interiorTranslation}
For $\xi \in \Lang$, denote by $\xi^\# \in \Langc$ the result of substituting each $\Box$ with $\Boxc$. 

1) If $\xi \in \Langc$, then $\slg{TLR_c} \vdash \xi \leftrightarrow \xi^\#$.

2) If $\slg{TLR}\vdash\xi$, then $\slg{TLR_c}\vdash \xi^\#$. 
\end{lemma}
\begin{proof}
1) Since $\xi^\#$ is obtained from $\xi$ by substituting each occurrence of $\Boxc\phi =\phi\wedge\Box\phi$ with $\phi\wedge \Boxc\phi$, it suffices to note that $\slg{S4} \vdash \phi \wedge\Boxc\phi \leftrightarrow \Boxc\phi$. 

\NewCase 2) By induction on the derivation of $\slg{TLR}\vdash\xi$.
If $\xi$ is an axiom of $\slg{K4}$ or is obtained by a rule of $\slg{K4}$ or is obtained by (Mon), it suffices to note that those axioms and rules are also among the axioms and rules of $\slg{TLR_c}$ with $\Boxc$ substituted for $\Box$. 
If $\xi=\gamma(\phi\wedge\Boxd\neg\phi,\top)\to\phi$ is an instance of (A8), then $\slg{S4} \vdash \phi^\#\wedge\Boxc\neg\phi^\#\to\bot$ and $\slg{TLR_c} \vdash \gamma(\bot,\top)\to\phi^\#$ by (A3), hence $\slg{TLR_c} \vdash \gamma(\phi^\# \wedge \Boxc\neg\phi^\#, \top)\to\phi^\#$ by (Mon).
Now suppose $\xi$ is an instance of (A0)--(A7). Since $\Box$ does not occur in those axioms (except as $\Boxc$), we can present $\xi$ as $\chi[p_1/\psi_1,\dotsc,p_k/\psi_k]$, where $\chi \in \Langc$ is itself an occurrence of (A0)--(A7). Since $\slg{TLR_c} \vdash \chi$, we have $\slg{TLR_c} \vdash \chi^\#$ by (1), hence  
$\slg{TLR_c} \vdash \chi^\#[p_1/\psi_1^\#,\dotsc,p_k/\psi_k^\#]=\xi^\#$.
\end{proof}

\begin{theorem}
The logic $\slg{TLR_c}$ is decidable, sound for the class of all topological spaces, and complete for the class of metric spaces.
\end{theorem}
\begin{proof}
Soundness of $\slg{S4}$ is well-known~\cite{BB07}. Soundness of (A0)--(A7) and (Mon) follows from Proposition~\ref{P:soundness}. If $\xi \in \Langc$ is valid over the class of metric spaces, then $\slg{TLR} \vdash \xi$ by Theorem~\ref{T:decidability}, whence $\slg{TLR_c} \vdash \xi$ by Lemma~\ref{L:interiorTranslation}. Decidability follows by Theorem~\ref{T:decidability}.
\end{proof}

\section{Final Remarks}

In this paper we axiomatized the logic of all topological spaces in the language with path-reachability $\gamma$ and the interior modality, but in the language enriched by the Cantor derivative we only addressed $T_1$ spaces.
As illustrated in Example~\ref{E:allTopologies}, the interaction of Cantor derivative with $\gamma$ in non-$T_1$ spaces has extra phenomena not caught by the axioms of this paper. To establish completeness for all spaces, it might be useful to enrich the tree-like construction of Proposition~\ref{P:networkTree} with ``jump'' edges isomorphic to the 2-point Sierpiński space.
Another problem, already mentioned in~\cite{BBCFG24}, is to axiomatize the $\gamma$-logic of the Euclidean plane.

\textbf{Acknowledgments.} The research on which this publication is based received financial support from project grant PID2023-149556NB-I00 and CEX2021-001169-M, funded by MICIU\slash{}AEI\slash{}10.13039\slash{}501100011033.
The work of the first author was supported by University of Barcelona collaboration grant, call 2025.4.FF.1.

\nocite{*}
\bibliographystyle{eptcs}
\bibliography{generic}

\end{document}